\documentclass[10pt, journal]{IEEEtran}

\usepackage{booktabs} 
\usepackage[utf8]{inputenc}
\usepackage{url}
\usepackage{amsthm}
\usepackage{amssymb}
\usepackage{algorithmicx, algorithm}
\usepackage[noend]{algpseudocode}
\usepackage{verbatim}
\usepackage{graphicx}
\usepackage{amssymb}
\usepackage{tikz}
\usepackage{listings}
\usepackage{booktabs}
\usepackage{array}
\usepackage{chngcntr}
\usepackage{url}
\usepackage{mathtools}
\usepackage{xcolor}
\usepackage{nomencl}
\usepackage{siunitx,xinttools,xintfrac}
\usepackage{lipsum}
\usepackage{subfig}
\usepackage{xspace}

\usepackage{pgfplots}
\usepackage{filecontents}
\usetikzlibrary{matrix}
\usepgfplotslibrary{groupplots}
\pgfplotsset{compat=newest}

\newcommand{\cspi}{\texttt{cLPI}\xspace}
\newcommand{\te}{\texttt{TE}\xspace}
\newcommand{\jsg}{\texttt{JSG}\xspace}
\newcommand{\fsi}{\texttt{FI}\xspace}
\newcommand{\ii}{\texttt{II}\xspace}
\newcommand{\tb}{\texttt{TB}\xspace}
\newcommand{\cpl}{\texttt{CPL}\xspace}

\begin{document}

\title{Length-Bounded Paths Interdiction in Continuous Domain for Network Performance Assessment}

\author{Lan N. Nguyen and My T. Thai 
\thanks{

Lan N. Nguyen and My T. Thai are with the Department of Computer and Information Science and Engineering, University of Florida, Gainesville, FL 32611 USA (email: lan.nguyen@ufl.edu, mythai@cise.ufl.edu).
}}

\maketitle

\begin{abstract}
Studying on networked systems, in which a communication between nodes is functional if their distance under a given metric is lower than a pre-defined threshold, has received significant attention recently.  In this work, we propose a metric to measure network resilience on guaranteeing the pre-defined performance constraint. This metric is investigated under an optimization problem, namely \textbf{Length-bounded Paths Interdiction in Continuous Domain} (\cspi), which aims to identify a minimum set of nodes whose changes cause routing paths between nodes become undesirable for the network service.


We show the problem is NP-hard and propose a framework by designing two oracles, \textit{Threshold Blocking} (\tb) and \textit{Critical Path Listing} (\cpl), which communicate back and forth to construct a feasible solution to \cspi with theoretical bicriteria approximation guarantees. Based on this framework, we propose two solutions for each oracle. Each combination of one solution to \tb and one solution to \cpl gives us a solution to \cspi. The bicriteria guarantee of our algorithms allows us to control the solutions's trade-off between the returned size and the performance accuracy. New insights into the advantages of each solution are further discussed via experimental analysis.

\end{abstract}

%
%



\IEEEpeerreviewmaketitle

\newtheorem{mydef}{Definition}
\newtheorem{lemma}{Lemma}
\newtheorem{theorem}{Theorem}
\newtheorem{Corollary}{Corollary}
\newtheorem{Observation}{Observation}
\algdef{SE}[DOWHILE]{Do}{doWhile}{\algorithmicdo}[1]{\algorithmicwhile\ #1}%
\DeclarePairedDelimiter\ceil{\lceil}{\rceil}
\DeclarePairedDelimiter\floor{\lfloor}{\rfloor}
\DeclarePairedDelimiter\norm{\lVert}{\rVert}
\newenvironment{skproof}{%
  \renewcommand{\proofname}{Proof overview}\proof}{\endproof}

\section{Introduction} \label{sec:introduction}

Components of a network do not have the same important level. There always exists a set of nodes or edges which plays more critical role than the others on assessing networks' performance. Literature has spent a significant effort on identifying such a set whose removal maximally damages a network's functionality. Most of early efforts used the connectivity metric, in which a connection between two nodes is functional if there exists a path connecting them \cite{garg1997primal, chawla2006hardness, svitkina2004min, dahlhaus1992complexity, dinh2015assessing, dinh2015network}. 
However, as modern networks evolved, purely relying on connectivity is no long sufficient to guarantee a networks' functionality or quality of services. Further, instead of removing nodes/edges, a change on components' behavior can downgrade the whole system's performance. For example, a congestion or traffic jams \cite{checkoway2011comprehensive, chen2018exposing} on some routers can significantly delay communication between end systems, downgrading their quality of services. 

Motivated by the above observations, recent researches turn their attention to network malfunction without damaging the connectivity. For example, Kuhnle et al. \cite{kuhnle2018network} studied the \texttt{LB-MULTICUT} problem, which aims for a minimum set of edges whose removal causes the shortest distance, in term of edge weights, between targeted pairs of nodes exceed a threshold. The threshold represents constraints for the networks in order to guarantee quality of services. By discarding the ``removal" flavour, Nguyen et al. \cite{nguyen2019network} extended this concept to introduce the \texttt{QoSD} problem, in which an edge weight can be varied with an amount of efforts, defined in the discrete domain, and the problem asks for a minimum amount of efforts for the same objective as in \texttt{LB-MULTICUT}. 


However, these existing works are all in the combinatorial optimization, which do not capture well the continuous settings. For example, in network routing, factors that impact network components' latency or packet loss rate include: traffic rate \cite{safaei2006software}, power of the interfering signal and noises \cite{alvizu2012hybrid}, denial-of-service attacks \cite{wangen2016cyber}. Those factors are quantified under continuous variables. 

Therefore, in this work, we take a further step on network performance assessment by introducing the \cspi{} problem as follows: Given a directed network $G=(V,E)$, a set $S$ of target pairs of nodes and a threshold $T$, each node $v \in V$ is associated with a monotone non-decreasing function $f_v: \mathbb{R}^\geq \rightarrow \mathbb{R}^\geq$, the \cspi problem asks for an impact vector $\mathbf{x} = \{x_v\}_{v \in V}$ with minimum $\norm{\mathbf{x}}$ such that any path $p$, connecting a pair in $S$, satisfies $\sum_{v \in p} f_v(x_v) \geq T$. Intuitively, $x_v$ represents the external impact's level to node $v$; while $f_v(x_v)$ quantifies $v$'s behaviors in response to the impact. $T$ represents the network constraint in order to guarantee quality of services, e.g. low latency. A solution $\mathbf{x}$ of \cspi can be used to assess network resilience to the external impact. 
Specifically, large $\norm{\mathbf{x}}$ indicates the network is resilient to external interference and able to maintain quality of service under extreme environment. Furthermore, a value of $x_v$ indicates the important level of node $v$ to the network desired functionality.

Solving \cspi with bounded performance guarantee is challenging, indeed. If for all path $p$, $\sum_{v \in p} f_v(x_v) \geq T$ exhibits convexity, then \cspi can be solved optimally by using ellipsoid method \cite{grotschel1981ellipsoid} with a polynomial feasible-check oracle. However, that is not always the case. Studies on network latency w.r.t impact factor like traffic rate shows $\sum_{v \in p} f_v(x_v) \geq T$ is not convex \cite{wangen2016cyber, feng2019feasibility}. That also rules out the possibility of applying any other Convex Optimization technique. In term of packet loss rate, the behaviors are even more complicated \cite{alvizu2012hybrid}. Indeed, we show that \cspi with general functions are NP-hard problem. Thus, in this work, we aim for a general solution that can applied on any monotone non-decreasing functions $f_v$. 


{\bf Contributions.} In addition to introduce the \cspi problem, the main contributions of this work are:

\begin{itemize}
    \item We propose a general framework for solving \cspi, separating tasks into two different oracles, called \textit{Critical Paths Listing} (\cpl) and \textit{Threshold Blocking} (\tb). \cpl's job is to restrict the amount of paths considered for finding a feasible solution of \cspi. \tb handles the task of finding $\mathbf{x}$, guaranteeing all paths, returned by \cpl, having lengths exceeded $T$. For each oracle, we design two algorithmic solutions. Different combination of any \cpl and \tb algorithms provides different performance theoretically and practically.
    \item All of our solutions have bicriteria approximation ratios, which could allow a user to control the trade-off between runtime versus accuracy.
    \item We extensively evaluate our solutions on real-world AS networks. The experiments show our algorithms outperforms existing solutions of special problems of \cspi in solution quality. We then shed a new insight on advantages of each algorithms.
\end{itemize}

\textbf{Organization}. The rest of the paper is organized as follows. Section \ref{sec:related} reviews literature related to our problem. In Section \ref{sec:problem}, we formally define the \cspi problem, discuss its challenges and overall framework of our solutions. Section \ref{sec:tb} presents two algorithms for the \tb oracle while the ones for \cpl are described in Section \ref{sec:cpl}. In Section \ref{sec:experiment}, practical analysis on algorithms' performance is provided. Finally, Section \ref{sec:conclusion} concludes the paper.

\section{Application and Related Work} \label{sec:related}

We first discuss a key application of \cspi in network performance assessment and next highlight the most relevant related work to \cspi. 

\subsection{\cspi in network performance assessment}

A routing protocol specifies how routers communicate with each other to distribute information that enables them to select routes between any two nodes on a computer network \cite{clausen2003optimized, waitzman1988distance, moy1998ospf}. The specific characteristics of routing protocols include the manner in which they avoid routing loops and select preferred routes, using information about hop costs. With the introduction of Software-Defined-Networking \cite{kreutz2014software,xia2014survey}, a hop cost can vary from different metrics, serving for different purposes of network administrators. 

The most common used metric for network vulnerability is network latency. Ideally, communication between hosts in the network is routed in the shortest path, weighted by latency of nodes. On the other hand, to guarantee quality of services (e.g. low latency) or avoid unexpected routing scheme (e.g. inter-continent routing with intra-traffic), a limit on network latency can be set so that the routing path has to have latency lower than a threshold. If there exists no routing path with total latency lower than the threshold, the network is considered to be undesirable for required services \cite{kuhnle2018network, nguyen2019network}. 

In the context of \cspi, to model the external impact to a hop latency, each node $v$ (e.g routers) in the network is associated with a function $d_v(x)$ where $x$ quantifies the impact (e.g. traffic rate, noise); and $d_v(x)$ measure the latency of $v$ with the impact $x$. Denote $T$ as the latency threshold. Studying \cspi helps identify the impact levels on nodes/edges that required to damage the networking quality of services, thus providing a useful metric for network design and assessment.

Beside latency, another routing metric can be used is packet loss probability. A routing path with high probability (say at least $90\%$) of successful delivery is preferred. Unlike latency, in term of packet loss probability, a simple trick needs to be applied. Let $\rho_v(x)$ denote the loss probability of a packet if going through node $v$ given the external impact amount $x$; and $P$ is the expected successful probability of a routing path. Then the network routing is not functional if for a routing path $p$, $\prod_{v \in p} (1-\rho_v(x)) \leq P$. This equation is adapted to \cspi as $\sum_{v \in p} -\ln (1-\rho_v(x)) \geq - \ln P$. 

\subsection{Existing Algorithms}

The early work on network resilience assessment with constraints on distance between node pairs are \texttt{LB-MULTICUT}, \cite{kuhnle2018network}, Critical Node Detection \cite{shen2013discovery, nguyen2013detecting}, Multicut \cite{garg1997primal, chawla2006hardness}. With the objective to make all pairs' distance to be at least $T$ or be disconnected, the problem asks for a minimum set of edges or nodes to be removed. One way to apply their solutions to \cspi is by introducing a step of discretization of function $f_v$. In the context of node removal, the cost of cutting a node $v$ is represented by value $x$ where $f_v(x) = T$. Our experimental results, unfortunately, shows solving \cspi by this method returns undesirable solutions in some cases. Other than that intuitive adoption, it is unclear how to convert the work of edge/node removal to the flavour of increase edge/node weight as an instance of \cspi. 

Without targeting for the edge/node removal, the \texttt{QoSD} problem introduces a discrete function $b_v: \mathbb{Z}^\geq \rightarrow \mathbb{Z}^+$ associated with each node $v$ of the network; and asks for a $\mathbf{x}$ in integer lattice that any path $p$ connecting a given node pairs has length exceeding $T$, i.e. $\sum_{v \in p} b_v(x_v) \geq T$ \cite{nguyen2019network}. One may think to discretize functions $f_v$ and directly adopt solutions of \texttt{QoSD} to solve \cspi. However, the discretization of $f_v$ is simply a work of taking an integer $x$ and returning the value $f_v(x \times \delta)$, where $\delta$ is called discretizing step. If $\delta$ is too large, the returned solution will be far from optimum due to discretization error; otherwise small $\delta$ creates significantly large inputs for \texttt{QoSD}, causing a burden on memory usage and undesirable runtime. Therefore, a solution, which can directly applied into continuous domain, is more desired.

\cspi can be modeled under a Constrained Optimization formulation that minimize $\sum_{v \in V} x_v$ with constraints $\sum_{v \in p} f_v(x_v) \geq T$ for all paths $p$ connecting pairs and $x_v \geq 0$
for all $v \in V$. 
The first constraint is to guarantee that all paths connecting target pairs have the length exceeding threshold $T$. Constrained Optimization is a classical problem, on which significant amount of works have been investigated, including (but not limited to) \cite{bertsekas2014constrained, gill2005snopt, runarsson2000stochastic}. However, a major concern of applying those solutions to \cspi is that a set of constraints is required to be known beforehand. In the case of dense network, the set of constraints reach to $\sum_{k=2}^n \binom{n}{k} k!$ paths and an ``infinite" period only for enumerating them. Furthermore, even we can list all constraints, those methods meet another obstacle that edge weight functions could be any function with complex behaviors. Existing methods can easily end up to local convergence trap without any performance guarantee. Therefore, a solution, which helps reduce burden of path listing while providing a performance ratio, is more desirable. That is a focus of our work.

\section{Preliminaries} \label{sec:problem}

\subsection{Problem Formulation}
In this part, we formally define the $\cspi$ problem and notations used frequently in our algorithms.  

We abstract the network using a directed graph $G=(V,E)$ with $|V|=n$ nodes and $|E| = m$ directed edges. Each node $v$ is associated with a function $f_v: \mathbb{R}^\geq \rightarrow \mathbb{R}^\geq$ which indicates the weight (e.g. latency, loss rate) of $v$ w.r.t an impact amount on $v$. In another word, if external impact of an amount of $x$ is put on $v$, the weight of node $v$ will become $f_v(x)$. $f_v$ is monotonically non-decreasing for all $v \in V$, which can be intuitively explained by: the more impact are put on $v$, the worse $v$ behaves (e.g. long latency, high loss rate).  

Given $V=\{v_1,...v_n\}$, we denotes the impact in form of a vector $\mathbf{x} = \{x_1,...x_n\}$ where $x_n$ is an impact on node $v_i$. For simplicity, we use the notation $v$ to present a node in $V$ and its index in $V$ also. So $x_v$ means the impact on node $v$, and the entry in $\mathbf{x}$ corresponding to $v$ also. The overall impact on all nodes, therefore, is $\norm{\mathbf{x}} = \sum_{v \in V} x_v$.


A path $p = \{u_0,u_1,...u_l\} \in G$ is a sequence of vertices such that $(u_{i-1},u_i) \in E$ for $i=1,..,l$. A path can also be understood as the sequence of edges $\{(u_0,u_1), (u_1,u_2),... (u_{l-1}, u_l)\}$. In this work, a path is used interchangeably as a sequence of edges or a sequence of nodes. 

Under an impact vector $\mathbf{x}$, the length of a path $p$ is denoted as $d_{\mathbf{x}}(p)$ where $d_{\mathbf{x}}(p) = \min\big( \sum_{v \in p} f_v(x_v), T \big)$. The $\min$ term is to bound a path's length by $T$. Since we only care about paths of length at most $T$, this bound does not impact our algorithms' results or the problem's generality.

We abuse the notation by also using $d$ to denote distance between two nodes in the network. To be specific, $d_\mathbf{x}(u,v)$ denotes distance between node $u$ and $v$ under $\mathbf{x}$, i.e $d_\mathbf{x}(u,v) = \min_{p = \{u,...,v\}} d_\mathbf{x}(p)$.

A \textit{single path} is a path that there exists no node who appears more than once in the path. Let $\mathcal{P}_i$ denote a set of simple paths connecting the pair $(s_i,t_i) \in S$ such that $\sum_{v \in p} f_v(0) < T$ for all $p \in \mathcal{P}_i$. Let $\mathcal{F} = \cup^k_{i=1} \mathcal{P}_i$, we call a path $p \in \mathcal{F}$ a \textit{feasible path} and $\mathcal{F}$ is a set of all feasible paths in $G$. A non-feasible path either connects no pair in $S$ or has an initial length exceed $T$.  $\cspi$ is formally defined as follows:

\begin{mydef} \textnormal{Length-bounded Paths Interdiction in Continuous Domain ($\cspi$).} Given an undirected graph $G=(V,E)$, a set $f = \{f_v: \mathbb{R}^\geq \rightarrow \mathbb{R}^\geq \}$ of node weight functions w.r.t impact on nodes and a target set of pairs of nodes $S = \{(s_1,t_1),...(s_k,t_k)\}$, determine an impact vector $\mathbf{x}$ with a minimum $\norm{\mathbf{x}}$ such that $d_\mathbf{x}(s_i,t_i) \geq T$ for all $(s_i, t_i) \in S$.
\end{mydef}





Let's look at several mathematical operators on vector space $\mathbb{R}^n$, which are used along the theoretical proofs of our algorithms. Given $\mathbf{x} = \{x_1,...x_n\}, \mathbf{y} = \{y_1,...y_n\} \in \mathbb{R}^n$, define:
\begin{align*}
\mathbf{x} + \mathbf{y} & = \{x_1+y_1,...x_n+y_n\} \\
\mathbf{x} \setminus \mathbf{y} &= \{\max(x_1 - y_1,0),... \max(x_n - y_n, 0)\}
\end{align*}

Moreover, we say $\mathbf{x} \leq \mathbf{y}$ if $x_v \leq y_v$ for all $v \in V$, the similar rule is applied to $<,\geq,>$. 

\begin{theorem}
 \cspi is an NP-hard problem
\end{theorem}

\begin{proof}
    We reduce \texttt{QoSD} to \cspi as follows: Given an instance of the \texttt{QoSD} problem, including a directed graph $G=(V,E)$ and a set of target pairs of nodes $S$, each node $v$ is associated with a monotone discrete functions $b_v: \mathbb{Z}^\geq \rightarrow \mathbb{Z}^+$. \texttt{QoSD} asks for a minimum $\norm{\mathbf{x}}$ that any path $p$ connecting a pair in $S$ has length exceeding $T$, i.e. $\sum_{v \in p} b_v(x_v) \geq T$.  
    
    We create an instance of \cspi by keeping $G,S,T$ and defining the node weight functions $\{f_v\}_v$ in continuous domain by letting $f_v(x) = b_v(\floor{x})$ for $x \in \mathbb{R}^\geq$ and $\forall v \in V$. $f_v$ is monotone non-decreasing function in continuous domain.
    
    It is trivial that each entry $x_v$ of an optimal solution of this \cspi instance should be an integer (or else we can replace $x_v$ by $\floor{x_v}$ and the \cspi's objective still go through). Since the optimal solution of \cspi contains all integers, it is also an optimal solution of \texttt{QoSD}. And vice versa, an optimal solution of \texttt{QoSD} is also an optimal solution of this \cspi instance.  Thus, \cspi is at least as hard as \texttt{QoSD}. And since \texttt{QoSD} has been proven to be NP-hard, \cspi is an NP-hard problem. 
\end{proof}

\subsection{General model of our solutions}


{\bf Properties of Performance Guarantees.} Given a problem instance with a threshold $T$, denote $\mathtt{OPT}$ as an optimal solution. We call an impact vector $\mathbf{x}$ is \textbf{$\varepsilon$-feasible} to \cspi iff under $\mathbf{x}$, the distance between each target pair is at least $T(1-\varepsilon)$. Our algorithms are bicriteria approximation algorithms, returning a $\varepsilon$-feasible solution $\mathbf{x}$ whose overall impact is bounded within a factor $O(\ln |\mathcal{F}|\varepsilon^{-1})$ of $\mathtt{OPT}$. $\varepsilon$ is treated as a trade-off between the algorithms' accuracy and returned $\norm{\mathbf{x}}$. To be specific, the smaller $\varepsilon$ is, the closer pairs' distances are to $T$ but the larger the returned solution is. $\varepsilon$ is adjustable, allowing users to control the trade-off as desired.  

{\bf General Framework.} Our solutions contain two separate oracles, called \textit{Threshold Blocking} (\tb) and \textit{Critical Paths Listing} (\cpl). These two oracles communicate back and forth with each other to construct a solution to \cspi, given an input instance of \cspi and a parameter $\varepsilon$. These two oracles are proposed to tackle two challenges of \cspi as stated before, to be specific:

\begin{itemize}
    \item \textit{Threshold Blocking} - a primary role of \tb is to solve a sub-problem of \cspi: Given a target set $\mathcal{P}$ of single paths and an initial impact vector $\mathbf{x}$, \tb aims to find $\mathbf{s}$ of minimum $\norm{\mathbf{s}}$ to $\mathbf{x}$ in order to make $d_{\mathbf{x} + \mathbf{s}}(p) \geq T$ for all $p \in \mathcal{P}$. For simplicity, we call this task \tb problem.
    \item \textit{Critical Paths Listing} - this oracle restricts the number of paths to be considered in the algorithm, thus significantly reducing the searching space and burdens on algorithms' runtime and memory for storage.
\end{itemize}

We propose multiple algorithms for each oracle. Specifically, we devise two algorithms for \cpl, which are \textit{Feasible Set Construction} and \textit{Incremental Interdiction}. To solve \tb, we develop two algorithms, called \textit{Threshold Expansion} and \textit{Jump Start Greedy}. Different combinations of \cpl and \tb algorithms provide different performances theoretically and experimentally. 

In general, the flow of our algorithms is:
\begin{enumerate}
    \item The algorithm starts with $x_v = 0$ for all $v \in V$.
    \item Given the current state of $\mathbf{x}$, by using a technique to restrict searching space, \cpl oracle searches for a set of critical paths $\mathcal{P}$, who are feasible paths and $d_\mathbf{x}(p) < T$ for all $p \in \mathcal{P}$.
    \item Then those paths along with a current state of $\mathbf{x}$ are given as an input for the \tb oracle, which then finds an additional budget $\mathbf{s}$ for $\mathbf{x}$ to make $d_{\mathbf{x} + \mathbf{s}}(p) \geq T$ for all $p \in \mathcal{P}$.
    \item The additional budget $\mathbf{v}$ is then used for \cpl to check the feasibility. If adding $\mathbf{s}$ makes $\mathbf{x}$  $\varepsilon$-feasible, the algorithm returns $\mathbf{x} + \mathbf{s}$ and terminates. Otherwise, $\mathbf{s}$ is used to drive the searching space of \cpl and find a new value for $\mathbf{x}$; then step (2) is repeated.
\end{enumerate}




\section{Threshold Blocking Oracle} \label{sec:tb}

In this section, we present two algorithms for Threshold Blocking (\tb) Oracle, called \textit{Threshold Expansion} (\te) and \textit{Jump Start Greedy} (\jsg). 

To recap, \tb receives a set $\mathcal{P}$ of feasible paths from \cpl, an impact vector $\mathbf{x}$. The objective of \tb is to find an additional vector $\mathbf{s} = \{s_1,...s_n\}$ with minimum $\sum_v s_v$ such that $d_{\mathbf{x} + \mathbf{s}}(p) \geq T$ for all $p \in \mathcal{P}$. 

Denote $\mathbf{s}^*$ as an optimal solution, i.e.
\begin{align*}
    \mathbf{s}^* = argmin_{\mathbf{s} : d_{\mathbf{x} + \mathbf{s}}(p) \geq T~\forall p \in \mathcal{P}} \norm{\mathbf{s}}
\end{align*}

The bicriteria guarantee of our algorithms originates from \tb's algorithms. Say in another way, instead of finding an exact solution, the desired accuracy $\varepsilon$ is given to the \tb oracle so \tb's algorithms find $\mathbf{s}$ such that $d_{\mathbf{x} + \mathbf{s}}(p) \geq T (1-\varepsilon)$ for all $p \in \mathcal{P}$. 



Denote $\langle v,x \rangle \in \mathbb{R}^n$ as a vector which receives value $x$ at entry $v$ and $0$ elsewhere.

Given a path set $P$, a vector $\mathbf{w}$ and a node $v$, let:
\begin{align*}
    r_{P,\mathbf{w},v}(x) = \sum_{p \in P} \big( d_{\mathbf{w} + \langle v,x \rangle } (p) - d_\mathbf{w}(p) \big)
\end{align*}

Intuitively, $r_{P,\mathbf{w},v}(x)$ measures the total increasing lengths, under an impact vector $\mathbf{w}$, of paths in $P$ by adding an amount $x$ to entry $v$ of $\mathbf{w}$.

\subsection{Threshold Expansion}

In general, \te works in rounds and in each round, \te set up a requirement on an amount to be added in each node. The requirements are relaxed after each round in order to allow new amount to be added; and the algorithm stops when guaranteeing the obtained solution $\mathbf{s}$ make $d_{\mathbf{x} + \mathbf{s}} (p) \geq T(1-\varepsilon)$ for all $p \in \mathcal{P}$. 


The requirement in each round of \te is in a form of a number $M$, which is initiated to be a large number. An amount $x$ to be added into $v$ guarantees $x = max\{x > 0 \mid \frac{r_{\mathcal{P}, \mathbf{x} + \mathbf{s}, v}(x)}{x} \geq M \}$. Intuitively, the condition $\frac{r_{\mathcal{P}, \mathbf{x} + \mathbf{s}, v}(x)}{x} \geq M$ is to ensure the additional amount is meaningful and significant in comparison with putting an impact on other nodes. Since there could be a wide range of $x$ that can satisfy $\frac{r_{\mathcal{P}, \mathbf{x} + \mathbf{s}, v}(x)}{x} \geq M$, the algorithm targets for the maximum $x$ because it helps the algorithm quickly reach to the feasible solution. After $x$ is added to entry $v$, the algorithm discards paths $p$ that $d_{\mathbf{x}+\mathbf{s}}(p) \geq T(1-\varepsilon)$ out of $\mathcal{P}$ since those paths have fulfilled the algorithm's target.

After considering adding impacts to all nodes with a constraint in term of $M$, the algorithm reduces the value of $M$ to be $(1-\epsilon) M$ with $\epsilon$ is a constant parameter inputted to the algorithm. The reduction in $M$ is to let new impact amounts be added into nodes. On the other hand, $\epsilon$ impacts the performance of the algorithm. Intuitively, the lower value of $\epsilon$ is, the better solution quality the algorithm can obtain but the longer running-time for the algorithm to terminate. The pseudocode of the algorithm is presented in Alg. \ref{alg:te_alg}.

\begin{algorithm}[t]
	\caption{Theshold Expansion}
	\label{alg:te_alg}
    \begin{flushleft}
    \textbf{Input} $G, \{f_v\}_v, \mathcal{P}, T, \epsilon, \varepsilon, \mathbf{x}$ \\
	\textbf{Output}: $\mathbf{s}$ that $d_{\mathbf{x} + \mathbf{s}}(p) \geq T(1-\varepsilon)$ for all $p \in \mathcal{P}$
    \end{flushleft}
    \begin{algorithmic}[1]
		\State $\mathbf{s} = \{0\}_v$
		\State $M = $ a large number. \label{line:M}
        \State $v \leftarrow $ the first node in $E$
        \While{$\mathcal{P}$ is not empty} \label{line:inner_it_greedy}
            	\State $\hat{x} = \max\big\{ x \geq 0 \mid \frac{r_{\mathcal{P}, \mathbf{x} + \mathbf{s}, v}(x)}{x} \geq M  \big\}$ \label{line:te_query}
                \State $\mathbf{s} = \mathbf{s} + \langle v, \hat{x} \rangle$
                \State Remove paths $p$ that $d_{\mathbf{x} + \mathbf{s}} (p) \geq T(1-\varepsilon)$ out of $\mathcal{P}$
            \If{$v$ is the last node in $V$}
                \State $M = (1-\epsilon) M$ \label{line:M_reduce}
                \State $v \leftarrow $ start over with the first node
            \Else
                \State $v \leftarrow$ the next node.
            \EndIf
        \EndWhile
    \end{algorithmic}
    \begin{flushleft}
    	\textbf{Return } $\mathbf{s}$
    \end{flushleft}
\end{algorithm}

\te's theoretical performance is obtained with an assumption that:  $M$ - initiated at line \ref{line:M} of Alg. \ref{alg:te_alg} - satisfies:
\begin{align}
    M \geq \frac{r_{\mathcal{P}, \mathbf{w}, v}(x)}{x} \mbox{ for all } \mathbf{w} \geq \mathbf{x}, v \in V \mbox{ and } x \geq 0 \label{equ:M_condition}
\end{align}

This assumption can be removed if $f_v$s are differentiable everywhere. In that case, we set $M$ as the following lemma.

\begin{lemma} \label{lemma:M_te} 
If $f_v$s are differentiable everywhere, by setting
\begin{align*}
 M = |\mathcal{P}| \times \max_{x \geq 0, v \in V, f_v(x) \leq T} \frac{\partial f_v}{ \partial x}
\end{align*}
the condition \ref{equ:M_condition} is satisfied.
\end{lemma}

\begin{proof}
We have:
\begin{align*}
    \frac{r_{\mathcal{P}, \mathbf{w}, v}(x)}{x} & = \sum_{p \in \mathcal{P}: v \in p} \frac{d_{\mathbf{w} + \langle v,x \rangle}(p) - d_\mathbf{w}(p)}{x} \\
    &\leq |\mathcal{P}| \max_{p \in \mathcal{P}: v \in p} \frac{d_{\mathbf{w} + \langle v,x \rangle}(p) - d_\mathbf{w}(p)}{x} \\
    &\leq |\mathcal{P}| \max_{x : f_v(x) \leq T} \frac{f_v(x) - f_v(w_e)}{x - w_e} \\
    &\leq |\mathcal{P}| \times \max_{x \geq 0, v \in V, f_v(x) \leq T} \frac{\partial f_v}{ \partial x}
\end{align*}
which completes the proof.
\end{proof}

 From now on, for simplicity, when we analyze the performance of \te at an iteration of the \textbf{while} loop of line \ref{line:inner_it_greedy} Alg. \ref{alg:te_alg}, we refer $M$, $\mathbf{s}$, $\mathcal{P}$ as their values at that iteration. 

Let's consider at an iteration of line \ref{line:inner_it_greedy} Alg. \ref{alg:te_alg}, denote $\mathbf{s}^o = \{s_v^o\}_v = \mathbf{s}^* \setminus \mathbf{s}$. We have the following lemma.

\begin{lemma} \label{lemma:te_no_concave}
$s_v^o = 0$ or $\frac{r_{\mathcal{P}, \mathbf{x} + \mathbf{s}, v} (s_v^o)}{s_v^o} < \frac{M}{1-\epsilon}$ for all $v \in V$.
\end{lemma}

\begin{proof}
This lemma is trivial at the time each node being first observed because of the condition \ref{equ:M_condition}. Therefore, we only consider at an arbitrary moment after $M$ has been reduced by line \ref{line:M_reduce}.

Assume there exists a node $v$ such that $s_v^o > 0$ and $\frac{r_{\mathcal{P}, \mathbf{x}+\mathbf{s}, v}(s_v^o)}{s^o_v} \geq \frac{M}{1-\epsilon}$. Consider the last time $v$ is observed and $\hat{x} = \max\big\{ x \geq 0 \mid \frac{r_{\mathcal{P}, \mathbf{x}+\mathbf{s}^\prime, v} (x)}{x} \geq M^\prime \big\}$ where $\mathbf{s}^\prime$ is $\mathbf{s}$ before adding $\hat{x}$ into $v$; $M^\prime = M$ if $v$ was last observed in the current round, otherwise $M^\prime = \frac{M}{1-\epsilon}$.


We have $\mathbf{s} \geq \mathbf{s}^\prime + \langle v, \hat{x} \rangle$ but $\mathbf{s}$ and $\mathbf{s}^\prime + \langle v, \hat{x} \rangle$ have the same value at entry $v$, thus for any $p \in \mathcal{P}$ that contains $v$:
\begin{align*}
    & d_{\mathbf{x} + \mathbf{s}^\prime + \langle v, \hat{x} + s^0_v\rangle}(p) - d_{\mathbf{x} + \mathbf{s}^\prime + \langle v, \hat{x}\rangle}(p) \\
    & \quad \geq d_{ \mathbf{x} + \mathbf{s} + \langle v, s^0_v\rangle}(p) - d_{\mathbf{x} + \mathbf{s}}(p)
\end{align*}

Therefore, 
\begin{align*}
    r_{\mathcal{P}, \mathbf{x} + \mathbf{s}^\prime + \langle v,\hat{x} \rangle, v}(s_v^o) &\geq r_{\mathcal{P}, \mathbf{x} + \mathbf{s}, v}(s_v^o) \\
    & \geq s_v^o \frac{M}{1-\epsilon}
\end{align*}
So:
\begin{align*}
    r_{\mathcal{P}, \mathbf{x} + \mathbf{s}^{\prime}, v}(\hat{x} + s_v^o) & = r_{\mathcal{P}, \mathbf{x} + \mathbf{s}^{\prime}, v}(\hat{x}) + r_{\mathcal{P}, \mathbf{x} + \mathbf{s}^{\prime} + \langle v,\hat{x} \rangle, v}(s_v^o) \\
    & \geq M^\prime (\hat{x} + s_v^o)  
\end{align*}

Then an amount of at least $\hat{x} + s_v^o$ should be added into $v$, which contradicts to assumption that $\hat{x}$ is the selected amount. 
\end{proof}

Lemma \ref{lemma:te_no_concave} allows us to bound the performance guarantee of \te, which is shown in the following theorem.

\begin{theorem} \label{theorem:te_approx}
    Given $G, \{f_v\}_v, \mathcal{P}, T, \varepsilon, \mathbf{x}$, if $\mathbf{s}$ is the additional impact vector returned by \te and $\mathbf{s}^*$ is the optimal vector to make $d_{\mathbf{x} + \mathbf{s}^*}(p) \geq T$ for all $p \in \mathcal{P}$, then:
    \begin{align*}
        \norm{\mathbf{s}} \leq \frac{\ln\big( |\mathcal{P}| \varepsilon^{-1} \big) + 1}{1 - \epsilon} \norm{\mathbf{s}^*}
    \end{align*}
\end{theorem}
\begin{proof}
Let us consider at an arbitrary iteration of the \textbf{while} loop at line \ref{line:inner_it_greedy},  node $v$ is being observed, and $\hat{x}$ is a selected amount to add into $v$ but has not been added to $\mathbf{s}$ yet. Again, denote $\mathbf{s}^o = \{s_v^o\}_v = \mathbf{s}^* \setminus \mathbf{s}$. Without lost of generality, let $\hat{x} >0$. From lemma. \ref{lemma:te_no_concave}, we have:
\begin{align*}
    \frac{r_{\mathcal{P}, \mathbf{x} + \mathbf{s}, v}(\hat{x})}{\hat{x}} \geq (1-\epsilon) \frac{r_{\mathcal{P}, \mathbf{x} + \mathbf{s}, v}(s_u^o)}{s_u^o}
\end{align*}
for all $u \in V$ that $s_u^o > 0$.

Denote $\mathbf{h}_u = \{s_w + \mathbf{1}_{w > u} s_w^o\}_w$ for all $u \in V$. As $\mathbf{h}_u \geq \mathbf{s}$ but they have the same value at entry $u$, we have:
\begin{align*}
    r_{\mathcal{P}, \mathbf{x} + \mathbf{h}_u, u}(s_u^o) \leq r_{\mathcal{P}, \mathbf{x} + \mathbf{s}, u}(s_u^o)
\end{align*}
Therefore,
\begin{align*}
    & \sum_{p \in \mathcal{P}} \Big( d_{\mathbf{x} + \mathbf{s} + \mathbf{s}^o}(p) - d_{\mathbf{x} + \mathbf{s}}(p) \Big) = \sum_{u \in V} r_{\mathcal{P}, \mathbf{x} + \mathbf{h}_u, u}(s_u^o) \\
    & \quad \leq \sum_{u \in V} r_{\mathcal{P}, \mathbf{x} + \mathbf{s}, u}(s_u^o) \leq \sum_{u \in V} \frac{s_u^o}{\hat{x} (1-\epsilon)} r_{\mathcal{P}, \mathbf{x} + \mathbf{s}, v} (\hat{x}) \\
    & \quad \leq \frac{\norm{\mathbf{s}^*}}{\hat{x}(1-\epsilon)} r_{\mathcal{P}, \mathbf{x}+\mathbf{s}, v}(\hat{x})
\end{align*}

Since $\mathbf{s} + \mathbf{s}^o \geq \mathbf{s}^*$, $d_{\mathbf{x} + \mathbf{s} + \mathbf{s}^o}(p) = T$ for all $p \in \mathcal{P}$.

Now, let's assume the algorithm terminates after adding impact amounts into nodes $L$ times, denote $\hat{x}_1, ... \hat{x}_L$ as an added amount at each times ($\norm{\mathbf{s}} = \sum_{t=1}^L \hat{x}_t$). Also, denote $\mathbf{s}_t$, $\mathcal{P}_t$ as $\mathbf{s}$, $\mathcal{P}$ before adding $\hat{x}_t$ at time $t$. We have:
\begin{align*}
    & \sum_{p \in \mathcal{P}_t} \Big( T - d_{\mathbf{x} + \mathbf{s}_t}(p) \Big) \\
    & \quad \leq \frac{\norm{\mathbf{s}^*}}{\hat{x}_t(1-\epsilon)} \sum_{p \in \mathcal{P}_t} \Big( d_{\mathbf{x} + \mathbf{s}_{t+1}}(p) - d_{\mathbf{x} + \mathbf{s}_t}(p) \Big)
\end{align*}
A simple transformation and the fact that $\mathcal{P}_{t+1} \subseteq \mathcal{P}_t$ gives us:
\begin{align*}
    &\sum_{p \in \mathcal{P}_{t+1}} \Big( T - d_{\mathbf{x} + \mathbf{s}_{t+1}}(p) \Big) \\
    & \quad \leq \Big( 1 - \frac{\hat{x}_t(1-\epsilon)}{\norm{\mathbf{s}^*}} \Big) \sum_{p \in \mathcal{P}_t} \Big( T - d_{\mathbf{x} + \mathbf{s}_t}(p) \Big)
\end{align*}
Therefore,
\begin{align*}
    & \sum_{p \in \mathcal{P}_{L-1}} \Big( T - d_{\mathbf{x} + \mathbf{s}_{L-1}}(p) \Big) \\
    & \quad \leq \sum_{p \in \mathcal{P}_0} \Big( T - d_{\mathbf{x}}(p) \Big)  \prod_{t=1}^{L-1} \Big( 1 - \frac{\hat{x}_t(1-\epsilon)}{\norm{\mathbf{s}^*}} \Big) \\
    & \quad \leq |\mathcal{P}| T \Big( 1 - \frac{\sum_t^{L-1} \hat{x}_t (1-\epsilon)}{\norm{\mathbf{s}^*}(L-1)} \Big)^{L-1} \\
    & \quad \leq e^{-\frac{\norm{\mathbf{s}_{L-1}} (1-\epsilon)}{\norm{\mathbf{s}^*}}} |\mathcal{P}| T
\end{align*}


On the other hand, $\sum_{p \in \mathcal{P}_{L-1}} \Big( T - d_{\mathbf{x} + \mathbf{s}_{L-1}}(p) \Big) \geq T \varepsilon$ since $d_{\mathbf{x} + \mathbf{s}_{L-1}}(p) < T(1-\varepsilon)$ for all $p \in \mathcal{P}_{L-1}$; and $\mathcal{P}_{L-1} \neq \emptyset$. Therefore:
\begin{align*}
    \norm{\mathbf{s}_{L-1}} \leq \norm{\mathbf{s}^*} \frac{\ln\Big( |\mathcal{P}| \varepsilon^{-1}\Big)}{1-\epsilon}
\end{align*}

Now, let consider the final update, we have:
\begin{align*}
    \hat{x}_L \leq \frac{\norm{\mathbf{s}^*}}{1-\epsilon} \frac{\sum_{p \in \mathcal{P}_{L-1}} \Big( d_{\mathbf{x} + \mathbf{s}_{L}}(p) - d_{\mathbf{x} + \mathbf{s}_{L-1}}(p) \Big)}{\sum_{p \in \mathcal{P}_{L-1}} \Big( T - d_{\mathbf{x} + \mathbf{s}_{L-1}}(p) \Big)}  \leq \frac{\norm{\mathbf{s}^*}}{1-\epsilon}
\end{align*}
Finally, we have
\begin{align*}
    \norm{\mathbf{s}} = \norm{\mathbf{s}_{L-1}} + \hat{x}_L \leq \norm{\mathbf{s}^*} \frac{\ln\Big( |\mathcal{P}| \varepsilon^{-1}\Big) + 1}{1-\epsilon}
\end{align*}
which completes the proof.
\end{proof}

\subsection{Jump Start Greedy}


In general, \jsg works in a greedy manner that iteratively adds an  impact amount to a node which maximizes  $\frac{r_{\mathcal{P}, \mathbf{x}+\mathbf{s}, v}(x)}{x}$. The problem is that there exists cases due to traits of the functions $f_v$s, the selected budget is $0$ and the algorithm falls into infinite loops. We call such situation ``zero trap". \jsg overcomes that challenge by introducing \textbf{Jump Start} step to escape the zero trap while keeping a reasonable theoretical performance guarantee. 


\jsg runs in multiple iterations and for each iteration:
\begin{itemize}
    \item \textbf{Step (1)}, for each node $v$, the algorithm finds a budget $\hat{x}_v$ that maximizes $\frac{r_{\mathcal{P}, \mathbf{x} + \mathbf{s}, v}(\hat{x}_v)}{\hat{x}_v}$. If $\hat{x}_v = 0$ (which typically happens when $f_v$ is concave), we do the \textit{jump start} by forcing the minimum amount added to $v$ has to be at least a value of $\beta = O(\norm{\mathbf{s}^*} / n)$ (how we obtain the value of $\beta$ will be described later). In that case, $\hat{x}_e = argmax_{x \geq \beta} \Big\{ \frac{r_{\mathcal{P}, \mathbf{x}+\mathbf{s}, v}(x)}{x} \Big\}$.
    \item \textbf{Step (2)}, the algorithm selects a node $v$ that maximizes $\frac{r_{\mathcal{P}, \mathbf{x}+\mathbf{s}, v}(\hat{x}_v)}{\hat{x}_v}$ and add $\hat{x}_v$ into $v$. The algorithm repeats to \textbf{step (1)} until $d_{\mathbf{x} + \mathbf{s}}(p) \geq T(1-\varepsilon)$ for all $p \in \mathcal{P}$.
\end{itemize}


The pseudo-code of \jsg is presented in Alg. \ref{alg:jsg_alg} and \jsg's performance guarantee is stated in the following theorem.

\begin{algorithm}[t]
	\caption{Jump Start Greedy}
	\label{alg:jsg_alg}
    \begin{flushleft}
    \textbf{Input} $G, \{f_v\}_v, \mathcal{P}, T, \varepsilon, \mathbf{x}$ \\
	\textbf{Output}: $\mathbf{s}$ that $d_{\mathbf{x} + \mathbf{s}}(p) \geq T(1-\varepsilon)$ for all $p \in \mathcal{P}$
    \end{flushleft}
    \begin{algorithmic}[1]
		\State $\mathbf{s} = \{0\}_v$
		\State $\beta = O(\norm{\mathbf{s}^*} / n) $ 
		\While{$\mathcal{P}$ is not empty} \label{line:while_jsg}
		    \For{each $v \in V$}
		        \State $\hat{x}_v = \max_x \frac{r_{\mathcal{P}, \mathbf{x} + \mathbf{s}, v}(x)}{x}$ \label{line:parallel_jsg}
		        \If{$\hat{x}_v = 0$ (Jump Start)}
		            \State $\hat{x}_v = \max_{x \geq \beta} \frac{r_{\mathcal{P}, \mathbf{x} + \mathbf{s}, v}(x)}{x}$
		        \EndIf
		    \EndFor
		    \State $v = argmax_{v \in V} \frac{r_{\mathcal{P}, \mathbf{x} + \mathbf{s}, v}(\hat{x}_v)}{\hat{x}_v}$
		    \State $\mathbf{s} = \mathbf{s} + \langle v, x_v \rangle$
		     \State Remove paths $p$ that $d_{\mathbf{x} + \mathbf{s}} (p) \geq T(1-\varepsilon)$ out of $\mathcal{P}$
		\EndWhile
    \end{algorithmic}
    \begin{flushleft}
    	\textbf{Return } $\mathbf{s}$
    \end{flushleft}
\end{algorithm}

\begin{theorem} \label{theorem:jsg_approx}
    Given $G, \{f_v\}_v, \mathcal{P}, T, \varepsilon, \mathbf{x}$ given to the \tb oracle. If $\mathbf{v}$ is the impact vector returned by \jsg and $\mathbf{v}^*$ is the optimal vector make $d_{\mathbf{x} + \mathbf{s}^*}(p) \geq T$ for all $p \in \mathcal{P}$, then
    \begin{align*}
        \norm{\mathbf{v}} \leq O\Big( \ln\big( |\mathcal{P}| \varepsilon^{-1} \big)\Big) \norm{\mathbf{v}^*}
    \end{align*}
\end{theorem}

\begin{proof}
Let's consider at a certain iteration of \textbf{while} loop (line \ref{line:while_jsg} Alg. \ref{alg:jsg_alg}), $\mathbf{s}$ is now under construction (not returned solution) and $\mathcal{P}$ is not empty. Again, denote $\mathbf{s}^o = \{s_v^o\}_v = \mathbf{s}^* \setminus \mathbf{s}$ and $\mathbf{h}_u = \{s_w + \mathbf{1}_{w > u} s_w^o\}_w$. From the proof of Theorem \ref{theorem:te_approx}, we have that $r_{\mathcal{P}, \mathbf{x} + \mathbf{h}_u, u}(s_u^o) \leq r_{\mathcal{P}, \mathbf{x} + \mathbf{s}, u}(s_u^o)$ and:
\begin{align*}
   \sum_{p \in \mathcal{P}} \Big( d_{\mathbf{x} + \mathbf{s} + \mathbf{s}^o}(p) - d_{\mathbf{x} + \mathbf{s}}(p) \Big) \leq \sum_{u \in V} r_{\mathcal{P}, \mathbf{x} + \mathbf{s}, u}(s_u^o)
\end{align*}

Due to monotonicity, $r_{\mathcal{P}, \mathbf{x} + \mathbf{s}, u}(s_u^o) \leq r_{\mathcal{P}, \mathbf{x} + \mathbf{s}, u}(s_u^o + \beta)$. We observe that: Even a node $v$ was forced to take jump start step or not, the selected amount $\hat{x}_v$ always satisfies $\frac{r_{\mathcal{P}, \mathbf{x} + \mathbf{s}, v}(\hat{x}_v)}{\hat{x}_v} \geq \frac{r_{\mathcal{P}, \mathbf{x} + \mathbf{s}, v}(x)}{x}$ for all $x \geq \beta$. Thus, let's assume $v$ is the selected node in this \textbf{while} iteration with the increasing impact amount of $\hat{x}_v$. Due to greedy selection, we have:
\begin{align*}
    \sum_{u \in V} r_{\mathcal{P}, \mathbf{x} + \mathbf{s}, u}(s_u^o) & \leq \sum_{u \in V} \frac{s_v^o + \beta}{\hat{x}_v} r_{\mathcal{P}, \mathbf{x} + \mathbf{s}, v}(s_v^o) \\
    & \leq \frac{\norm{\mathbf{s}^*} + \beta n}{\hat{x}_v} r_{\mathcal{P}, \mathbf{x} + \mathbf{s}, v}(s_v^o)
\end{align*}

Now, let's assume the algorithm terminates after adding impact amounts into nodes $L$ times, denote $\hat{x}_1, ... \hat{x}_L$ as an added amount at each times ($\norm{\mathbf{s}} = \sum_{t=1}^L \hat{x}_t$). Also, denote $\mathbf{s}_t$, $\mathcal{P}_t$ as $\mathbf{s}$, $\mathcal{P}$ before adding $\hat{x}_t$ at time $t$. Using the same transformation as in proof of \te, we obtain the resursion relationship between $\mathbf{s}_t, \mathcal{P}_t$ as follows:
\begin{align*}
    &\sum_{p \in \mathcal{P}_{t+1}} \Big( T - d_{\mathbf{x} + \mathbf{s}_{t+1}}(p) \Big) \\
    & \quad\leq \Big( 1 - \frac{\hat{x}_t}{\norm{\mathbf{s}^*} + \beta n} \Big) \sum_{p \in \mathcal{P}_t} \Big( T - d_{\mathbf{x} + \mathbf{s}_t}(p) \Big)
\end{align*}
Using the same technique as in \te to discarding the terms from round $t=1$ to $L-1$, we have
\begin{align*}
    \norm{\mathbf{s}} \leq \big( \norm{\mathbf{s}^*} + \beta n \big) O(\ln |\mathcal{P}|\varepsilon^{-1})
\end{align*}
The theorem follows given the fact that $\beta = O(\norm{\mathbf{s}^*} / n)$.
\end{proof}

Now the only question left is how to identify $\beta = O(\norm{\mathbf{s}^*} / n)$. The trivial answer is $\beta = 0$ but that does not help on the jump start step. To find a more reasonable lower bound of the optimal solution $\norm{\mathbf{s}^*}$, we have the following lemma.

\begin{lemma} \label{lemma:jsg_lower_opt}
Given a impact vector $\mathbf{x}$ such that there exists $p \in \mathcal{P}$, $d_\mathbf{x}(p) < T$, there exist $\sigma > 0$ such that with $\mathbf{w}(\sigma) = \{\sigma\}_v$, $d_{\mathbf{x} + \mathbf{w}(\sigma)}(p) < T$ and $\norm{\mathbf{s}^*} \geq \sigma$
\end{lemma}

\begin{proof}
 The first statement is trivial, so we will focus on the second statement. We have $d_{\mathbf{x} + \mathbf{s}^*}(p) > d_{\mathbf{x} + \mathbf{w}(\sigma)}(p)$. Thus there should exist at least one entry in $\mathbf{s}^*$ that is at least $\sigma$. So $\norm{\mathbf{s}^*} \geq \sigma$, which completes the proof.
\end{proof}




As $d_{\mathbf{x} + \mathbf{w}(\sigma)}(p)$ is monotone increasing w.r.t $\sigma$, we use binary search to find $\sigma$ and set $\beta = \frac{\sigma}{n}$.

\section{Critical Path Listing Oracle} \label{sec:cpl}

In this section, we present two algorithms for the \cpl oracle, which are \textit{Incremental Interdiction} (\ii) and \textit{Feasible Set Interdiction} (\fsi). \cpl's role is to reduce searching space when constructing the returned solution $\mathbf{x}$. \cpl works as a \textbf{backbone} for the overall process of finding $\mathbf{x}$, in which it receives \cspi's input, then communicates back and forth with \tb to construct $\mathbf{x}$ and returns $\mathbf{x}$ when $\mathbf{x}$ guarantees  $d_\mathbf{x}(s,t) \geq T(1 - \varepsilon)$ for all $(s,t) \in S$.

\subsection{Incremental Interdiction}

In general, this algorithm works in rounds; and in each rounds, impact amounts are added into nodes to guarantee a set of feasible paths getting length exceeding $T(1-\varepsilon)$. A set of paths are different and disjoint in each round. And to make all paths of that set have length exceed $T(1-\varepsilon)$, \ii calls the \tb oracle to find an additional impact vector to its current vector $\mathbf{x}$. The algorithm iterates until finding no feasible paths of length less than $T(1-\varepsilon)$. 

A set of paths in each round contains $k$ shortest paths connecting each pair of $S$ under its current impact vector $\mathbf{x}$. $k$ is a constant parameter inputted for the algorithm. Intuitively, $k$ is desired to be neither too large or too small. Large $k$ bring burdens on running time to find those shortest paths and memory to store them. On the other hand, small $k$ does not bring sufficient exposures for critical nodes, who appear frequently on paths connecting pairs in $S$ and are the ones the algorithm should target to put impact on. The pseudocode is presented in Alg. \ref{alg:incremental}. 



\begin{algorithm}[t]
	\caption{Incremental Interdiction}
    \label{alg:incremental}
	\begin{flushleft}
	\textbf{Input} $G, \{f_e\}_e, T, \varepsilon, S, \tb$ \\
	\textbf{Output} $\mathbf{x}$
	\end{flushleft}
    \begin{algorithmic}[1]
    	\State $\mathbf{x} = \{0\}$
		\While{$ \exists (s,t) \in S$ that $d_\mathbf{x}(s,t) < T(1 - \varepsilon)$} \label{line:ii_outer_it} 
        	\State $\mathcal{P} = \emptyset$
        	\For{each pair $(s,t) \in S$ that $d_\mathbf{x}(s,t) < T (1- \varepsilon)$}
        	\State $K = k$ shortest paths from $s$ to $t$ under $\mathbf{x}$
        	\State Remove paths $p$ that $d_\mathbf{x}(p) \geq T(1-\varepsilon)$ out of $K$
        	   \State $\mathcal{P} = \mathcal{P} \cup K$ 
        	\EndFor
            \State $\mathbf{s} = $ run \tb oracle with input $G,\{f_v\}_v,\mathcal{P},T, \varepsilon, \mathbf{x}$ 
            \State $\mathbf{x} = \mathbf{x} + \mathbf{s}$
        \EndWhile
     \end{algorithmic}
     \begin{flushleft}
     	\textbf{Return $\mathbf{x}$}
     \end{flushleft}
\end{algorithm}

Denote $t$ as the number of outer rounds (line \ref{line:ii_outer_it} Alg. \ref{alg:incremental}) \ii ran before terminating. \ii's theoretical performance guarantee is stated in the following theorem.

\begin{theorem}
 Given an instance $G, \{f_v\}_v, S, T$ of the \cspi problem and a \tb oracle, if $\mathbf{x}$ is an output of \ii and $\mathbf{x}^*$ is the optimal solution to the \cspi's instance, then
 \begin{align*}
     \norm{\mathbf{x}} \leq \norm{\mathbf{x}^*} O(t \ln \frac{|\mathcal{F}| \varepsilon^{-1}}{t})
 \end{align*}
\end{theorem}

\begin{proof}
    Denote $\mathbf{x}_i$ and $\mathcal{P}_i$ as $\mathbf{x}$ and $\mathcal{P}$ obtained at iteration $i$ of the loop at line \ref{line:ii_outer_it}. From approximation guarantee of the \tb oracle, we have:
    \begin{align*}
     \norm{\mathbf{x}_i \setminus \mathbf{x}_{i-1}} \leq \norm{\mathbf{x}^*} O(\ln(|\mathcal{P}_i|\varepsilon^{-1}))
    \end{align*}
    Therefore:
    \begin{align}
        \norm{\mathbf{x}} & = \sum_{i=1}^t \norm{\mathbf{x}_i \setminus \mathbf{x}_{i-1}} \leq \norm{\mathbf{x}^*}{} \sum_{i=1}^t O(\ln(|\mathcal{P}_i|\varepsilon^{-1})) \\
        & = \norm{\mathbf{x}^*} O(\ln \prod_{i=1}^t |\mathcal{P}_i| + t \ln  \varepsilon^{-1}) \\
        & \leq \norm{\mathbf{x}^*} O(\ln \Big( \frac{\sum_{i=1}^t |\mathcal{P}_i|}{t} \Big)^t) + O(t \ln \varepsilon^{-1}) \label{equ:ii_am_gm}\\
        & \leq \norm{\mathbf{x}^*} O(t \ln \frac{|\mathcal{F}|\varepsilon^-1}{t}) \label{equ:ii_all}
    \end{align}
    The inequality \ref{equ:ii_am_gm} is from AM-GM inequality while \ref{equ:ii_all} is from the fact that $L_i$s are disjoint sets of paths. Thus $\sum_{i=1}^t |\mathcal{P}_i| \leq |\mathcal{F}|$, which completes the proof.
\end{proof}




\subsection{Full Set Interdiction} \label{sec:igta}

In general, \fsi aims to construct a set $\mathcal{P}$ of feasible paths, which is a subset of $\mathcal{F}$ but, if being used as an input for \tb, can return $\mathbf{s}$ that is also an $\varepsilon$-feasible solution of \cspi. 

Different to \ii, which incrementally adds impact to interdict disjoint sets of feasible paths, \fsi aggregates all found path sets into a big one set called $\mathcal{P}$; and reset the impact vector $\mathbf{x}$ in order to find a new vector that can simultaneously interdict all paths in $\mathcal{P}$. A new path set is found by $k$ shortest paths algorithm with a same motive as \ii. The algorithm terminates when the output $\mathbf{s}$ of the \tb oracle with input $\mathcal{P}$ is also $\varepsilon$-feasible to \cspi. The pseudocode is presented in Alg. \ref{alg:iterative} and \fsi's performance guarantee is presented by the following theorem.

\begin{theorem} \label{theorem:ic_approx}
 Given an instance $G, \{f_v\}_v, S, T$ of the \cspi problem and a \tb oracle, if $\mathbf{x}$ is an output of \fsi and $\mathbf{x}^*$ is the optimal solution to the \cspi's instance, then
 \begin{align*}
     \norm{\mathbf{x}} \leq \norm{\mathbf{x}^*} O(\ln |\mathcal{F}| \varepsilon^{-1})
 \end{align*}
\end{theorem}

\begin{proof}
    Without loss of generality, let $\mathcal{P}$ denote as the final path sets inputted to \tb in the final iteration. From performance guarantee of \tb, we have that:
    \begin{align*}
         \norm{\mathbf{x}} \leq \norm{\mathbf{x}^*} O(\ln |\mathcal{P}| \varepsilon^{-1})
    \end{align*}
    The theorem trivially follows since there is no duplicated path in $\mathcal{P}$ and $\mathcal{P} \subseteq \mathcal{F}$.
\end{proof}

Although \fsi shows to have a better performance guarantee than \ii, in term of memory complexity, it could take \fsi $O(|\mathcal{F}|)$ to store $\mathcal{P}$ while \ii only takes $O(|S|k)$. That is the trade-off between those two algorithms and it will be shown in more detail in our experiments.

\begin{algorithm}[t]
	\caption{Full Set Interdiction}
    \label{alg:iterative}
	\begin{flushleft}
	\textbf{Input} $G, \{f_e\}_e, T, \varepsilon, S, \tb$ \\
	\textbf{Output} $\mathbf{x}$
	\end{flushleft}
    \begin{algorithmic}[1]
    	\State $\mathcal{P} = \emptyset, \mathbf{x} = \{0\}_v$
		\While{$ \exists (s,t) \in S$ that $d_\mathbf{x}(s,t) < T(1 - \varepsilon)$ in $G$} \label{line:outer_it} 
        	\For{each pair $(s,t) \in S$ that $d_\mathbf{x}(s,t) < T (1- \varepsilon)$}
        	    \State $K = k$ shortest paths from $s$ to $t$ under $\mathbf{x}$
        	    \State Remove paths $p$ that $d_\mathbf{x}(p) \geq T(1-\varepsilon)$ out of $K$
        	   \State $\mathcal{P} = \mathcal{P} \cup K$ 
        	 \EndFor
		    \State $\mathbf{x} = \{0\}$
            \State $\mathbf{s} = $ run \tb oracle with input $G,\{f_v\}_v,\mathcal{P},T, \varepsilon, \mathbf{x}$ 
            \State $\mathbf{x} = \mathbf{s}$
        \EndWhile
     \end{algorithmic}
     \begin{flushleft}
     	\textbf{Return $\mathbf{x}$}
     \end{flushleft}
\end{algorithm}

\section{Experimental Analysis} \label{sec:experiment}

In this section, we run simulation on network data sets to evaluate performance of different combination between algorithms of the \cpl and \tb oracle. We compare our algorithms' performance to several methods modified from existing solutions to adapt to the context of \cspi. The results show our algorithms outperform existing methods in most cases. We further investigate advantages of each algorithm to reveal some insights on use cases of each technique. 

\subsection{Experimental Settings}

We run experiments on a router network, collected from SNAP \cite{snapnets} dataset. The network is constructed as a communication network of who-talks-to-whom from the BGP (Border Gateway Protocol) logs. The network is undirected, containing 6474 nodes and 13895 undirected links connecting nodes.

Critical traffics are randomly sampled from pairs of end hosts in the networks. That critical traffics forms the set $S$ as an input to \cspi.

Due to lack of dataset information, for each experiment, we let $f_v$ be identical for all $v$, and be one of the following:
\begin{itemize}
    \item $f_v(x) = O(x^2)$ - a convex function in order to simulate the relation between external impacts to a router latency.
    \item  $f_v(x) = O(\log x)$ - a concave function to simulate the relation between external impacts to packet drop/loss rate of a router.
    \item $f_v(x) = O(x)$ - a linear function to compare our algorithms' solution quality to an optimal solution, which can be found by using linear programming.
    \item $f_v(x) = O(\floor{x})$ - a step function to compare our algorithms' performance with an existing discrete method.
\end{itemize}

We compare our algorithms with the following methods:
\begin{itemize}
    \item \texttt{CUT} - this method is adapted from \cite{kuhnle2018network}. In general, the method works in an ``all-or-nothing" manner that an impact amount put into a node is either $0$ or $\min\{x \mid f_v(x) \geq T\}$. That amount guarantees any path containing that node will have length at least $T$.
    \item \texttt{DISCRETE} - this method discretizes the functions $f_v$s as follows. If $f_v$ is a step function, the amount put into a node is a positive integer. Otherwise, the amount put into a node is among $0, x, 2x, 3x$ where $x = \min\{x \mid f_v(3x) \geq T\}$. The method then apply the \texttt{QoSD} algorithm to solve the discretized instance.
    \item \texttt{OPT} - this method is only applied when $f_v$ is a linear or step function. We use CPLEX \cite{cplex2009v12} to optimally solve the linear programming modelling the \tb oracle and combine it with \fsi in \cpl to obtain the optimal solution to \cspi.
\end{itemize}

With our algorithms, the most time-consuming part is on finding global optimum of univariate functions, for example $\max_x \frac{r_{\mathcal{P}, \mathbf{x} + \mathbf{s}, v}(x)}{x}$ in \jsg. As ``what is the best technique to find global optimum?" \cite{calvin2012adaptive,aaid2017new} is still an open question, we measured the runtime of our algorithms in term of how many times they have to query for finding global optimum of a univariate function. 

Finally, in the \cpl oracle, we set $k=20$, which - in our experiment - balances the trade-off on running time to find $k$ shortest path and the exposure of critical nodes. In \tb, with \te, we initially set $M = 1000000$ if the function is non-differentiable (e.g. step function). $\varepsilon = 0.1$ otherwise stated.

\textit{We only present representative experimental results. Other results with similar behaviors are excluded.}

\subsection{Results}

\subsubsection{How algorithms perform with various $T$?}

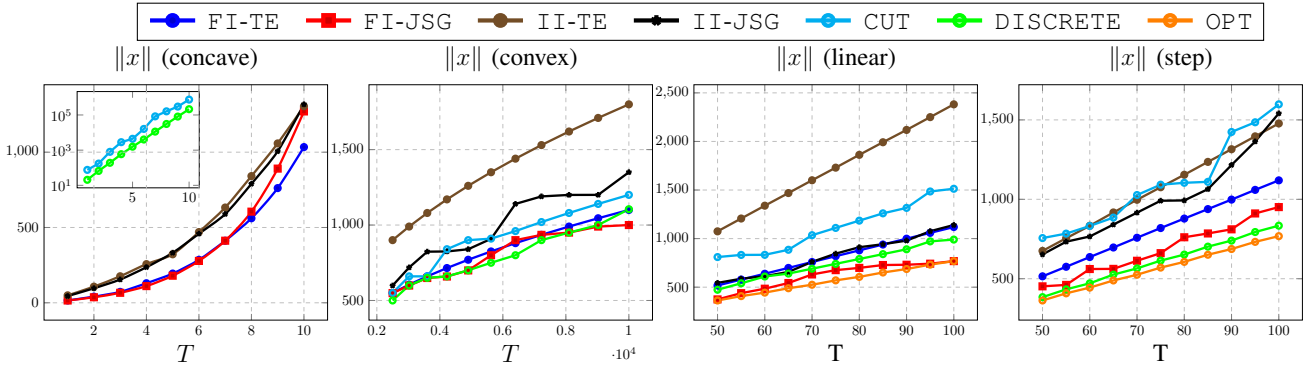
\begin{figure*}[t]
\begin{tikzpicture}[yscale=0.55, xscale=0.55]
    \begin{groupplot}[group style={group size= 4 by 1}]
        \begin{axis}[
                width=0.25\textwidth,ymode=log,
                xshift=0.8cm,
                yshift=3.2cm,
                every axis plot/.append style={ultra thick}
        ]
                \addplot[cyan,mark=o] table [x=T, y=CUT, col sep=comma] {data/concave_T/solution.csv};
                \label{plot:cut}
                \addplot[green,mark=o] table [x=T, y=DISCRETE, col sep=comma] {data/concave_T/solution.csv};
                \label{plot:discrete}
        \end{axis}
        \nextgroupplot[title={$\norm{x}$ (concave)},xlabel={$T$}, grid style=dashed, grid=both, grid style={line width=.1pt, draw=gray!10}, major grid style={line width=.2pt,draw=gray!50}, every axis plot/.append style={ultra thick, smooth},title style={font=\LARGE}, label style={font=\LARGE}]
                \addplot table [x=T, y=IC-TE, col sep=comma] {data/concave_T/solution.csv}; \label{plot:ic_te}
                \addplot table [x=T, y=IC-JSG, col sep=comma] {data/concave_T/solution.csv};
                \label{plot:ic_jsg}
                \addplot table [x=T, y=II-TE, col sep=comma] {data/concave_T/solution.csv};
                \label{plot:ii_te}
                \addplot table [x=T, y=II-JSG, col sep=comma] {data/concave_T/solution.csv};
                \label{plot:ii_jsg}
                \coordinate (top) at (rel axis cs:0,1);
        \nextgroupplot[title={$\norm{x}$ (convex)},xlabel={$T$},grid style=dashed, grid=both, grid style={line width=.1pt, draw=gray!10}, major grid style={line width=.2pt,draw=gray!50}, every axis plot/.append style={ultra thick},title style={font=\LARGE}, label style={font=\LARGE}]
                \addplot table [x=T, y=IC-TE, col sep=comma] {data/convex_T/solution.csv};
                \addplot table [x=T, y=IC-JSG, col sep=comma] {data/convex_T/solution.csv};
                \addplot table [x=T, y=II-TE, col sep=comma] {data/convex_T/solution.csv};
                \addplot table [x=T, y=II-JSG, col sep=comma] {data/convex_T/solution.csv};
                \addplot[cyan,mark=o] table [x=T, y=CUT, col sep=comma] {data/convex_T/solution.csv};
                \addplot[green,mark=o] table [x=T, y=DISCRETE, col sep=comma] {data/convex_T/solution.csv};
        \nextgroupplot[title={$\norm{x}$ (linear)},xlabel={T}, grid style=dashed, grid=both, grid style={line width=.1pt, draw=gray!10}, major grid style={line width=.2pt,draw=gray!50}, every axis plot/.append style={ultra thick},title style={font=\LARGE}, label style={font=\LARGE}]
                \addplot table [x=T, y=IC-TE, col sep=comma] {data/linear_T/solution.csv};
                \addplot table [x=T, y=IC-JSG, col sep=comma] {data/linear_T/solution.csv};
                \addplot table [x=T, y=II-TE, col sep=comma] {data/linear_T/solution.csv};
                \addplot table [x=T, y=II-JSG, col sep=comma] {data/linear_T/solution.csv};
                \addplot[cyan,mark=o] table [x=T, y=CUT, col sep=comma] {data/linear_T/solution.csv};
                \addplot[green,mark=o] table [x=T, y=DISCRETE, col sep=comma] {data/linear_T/solution.csv};
                \addplot[orange,mark=o] table [x=T, y=OPT, col sep=comma] {data/step_T/solution.csv}; \label{plot:opt}
        \nextgroupplot[title={$\norm{x}$ (step)},xlabel={T}, grid style=dashed, grid=both, grid style={line width=.1pt, draw=gray!10}, major grid style={line width=.2pt,draw=gray!50}, every axis plot/.append style={ultra thick},title style={font=\LARGE}, label style={font=\LARGE}]
                \addplot table [x=T, y=IC-TE, col sep=comma] {data/step_T/solution.csv};
                \addplot table [x=T, y=IC-JSG, col sep=comma] {data/step_T/solution.csv};
                \addplot table [x=T, y=II-TE, col sep=comma] {data/step_T/solution.csv};
                \addplot table [x=T, y=II-JSG, col sep=comma] {data/step_T/solution.csv};
                \addplot[cyan,mark=o] table [x=T, y=CUT, col sep=comma] {data/step_T/solution.csv};
                \addplot[green,mark=o] table [x=T, y=DISCRETE, col sep=comma] {data/step_T/solution.csv};
                \addplot[orange,mark=o] table [x=T, y=OPT, col sep=comma] {data/step_T/solution.csv};
                \coordinate (bot) at (rel axis cs:1,0);
    \end{groupplot}
\path (top|-current bounding box.north)--
      coordinate(legendpos)
      (bot|-current bounding box.north);
\matrix[
    matrix of nodes,
    anchor=south,
    draw,
    inner sep=0.2em,
    draw,
    nodes={anchor=center},
  ]at([xshift=-3cm]legendpos)
  {
    \ref{plot:ic_te}& \fsi-\te&[5pt]
    \ref{plot:ic_jsg}& \fsi-\jsg&[5pt]
    \ref{plot:ii_te}& \ii-\te&[5pt]
    \ref{plot:ii_jsg}& \ii-\jsg&[5pt]
    \ref{plot:cut}& \texttt{CUT}&[5pt]
    \ref{plot:discrete}& \texttt{DISCRETE} &[5pt]
    \ref{plot:opt}& \texttt{OPT} \\};
\end{tikzpicture}
\caption{Algorithms' returned solution with various $T$}
 	\label{fig:solution_T}
\end{figure*}

In the first set of experiments, we varied values of $T$ to observe how different algorithms performed. Figure \ref{fig:solution_T} displays $\norm{\mathbf{x}}$ returned by our algorithms in comparison with \texttt{CUT}, \texttt{DISCRETE} and \texttt{OPT} (only when $f_v$ is a linear or step function). 

In the concave case, we observe that our algorithms outperformed existing methods by a huge margin. Existing methods were totally undesirable in this case as their required impact were approximately 100 times worse than ours. This can be explained by: with the concave function, the contribution of impacts to the function expose diminishing return property, i.e. the function's gain becomes insignificant as input impact grows. That exposed the weakness of discretization steps in \texttt{CUT} and \texttt{DISCRETE} as a discretized impact's contribution is incomparable to the invested amount. 

On the other hand, our algorithms involving \fsi as the \tb oracle returns comparable solution quality to \texttt{OPT} and \texttt{DISCRETE} in non-concave functions. With non-concavity, the function's gain benefits when input impact increases. Critical nodes, which appear frequently on feasible paths connecting pairs in $S$, are tended to received large impact amount. Therefore, we observed \texttt{FI-JSG} and \texttt{DISCRETE} behaves almost similarly; and returns solution close to \texttt{OPT} in linear and step cases. Although our algorithms involving \ii returns solution larger than \fsi, they have advantages in running time and memory, which will be shown in the next parts. 

\subsubsection{How our algorithms' number of queries change with various $T$?}

In this experiment, we measured the number of queries each of our algorithms takes to solve a \cspi instance. Just to recall, a query is counted as a call to find global optimal of a univariate function. In algorithms involving \jsg, a query is equivalent to finding $\max_x \frac{r_{\mathcal{P}, \mathbf{x} + \mathbf{s}, v}(x)}{x}$ (line \ref{line:parallel_jsg} Alg. \ref{alg:jsg_alg}). In the ones involving \te, a query is counted as a call to find $ \max\big\{ x \geq 0 \mid \frac{r_{\mathcal{P}, \mathbf{x} + \mathbf{s}, v}(x)}{x} \geq M  \big\}$ (line \ref{line:te_query} Alg. \ref{alg:te_alg}). Figure \ref{fig:query_T} shows the numbers of queries taken by each algorithm in various $T$ and different impact functions.

From Figure \ref{fig:query_T}, we can see that our algorithm involving \ii totally outperformed the ones with \fsi in term of queries. For example, with concave cases, with a same \tb method, algorithms involving \fsi tends to take 100 times more queries than the one with \ii. With convex and step cases, this number is around 2-3 and it is around 5 in linear cases. This can be explained by the fact that \ii works in an incremental manner, in which impact amounts are accumulated when a new feasible paths - whose lengths have not satisfied the problem constraints - are found. Thus each query of algorithm involving \ii play a role, even insignificant, in constructing the final solution. Meanwhile \fsi resets its impact vector if new unsatisfactory feasible paths are found. Thus queries used before resetting the vector become wasted.

In comparison between algorithms of the \tb oracle, it can be seen that \te performed better in concave and linear cases while in convex and step, \jsg is the better one. That can be explained as follows: due to the trait of concave and linear functions, \jsg's query always returns an amount equal to the jump start step, i.e. $\beta$. Thus the algorithm required multiple queries to reach satisfactory amount. In contrast, the query $\max\big\{ x \geq 0 \mid \frac{r_{\mathcal{P}, \mathbf{x} + \mathbf{s}, v}(x)}{x} \geq M  \big\}$ of \te
 can reach to a larger amount in comparison with a jump start step. However, that situation does not happen when convexity is exposed. With convex functions, impact amounts are invested only on several nodes, which exactly is how \jsg behaves. Meanwhile, \te adds impact amounts to nodes sequentially, which makes \te's impact scattered and unnecessary on some nodes.

However, there is an interesting fact about \te:  \te's number of queries does not depends on $T$ in non-concave cases. That is the reason why \te's number of queries are constant in those cases as shown in Figure. \ref{fig:query_T}. That can be intuitively explained by that: given a set of paths $P$ which share a common node $v$, the way \te increases impact amount on $v$ by query $\max\big\{ x \geq 0 \mid \frac{r_{P, \mathbf{x} + \mathbf{s}, v}(x)}{x} \geq M  \big\}$ does not get affected by $T$'s value. 

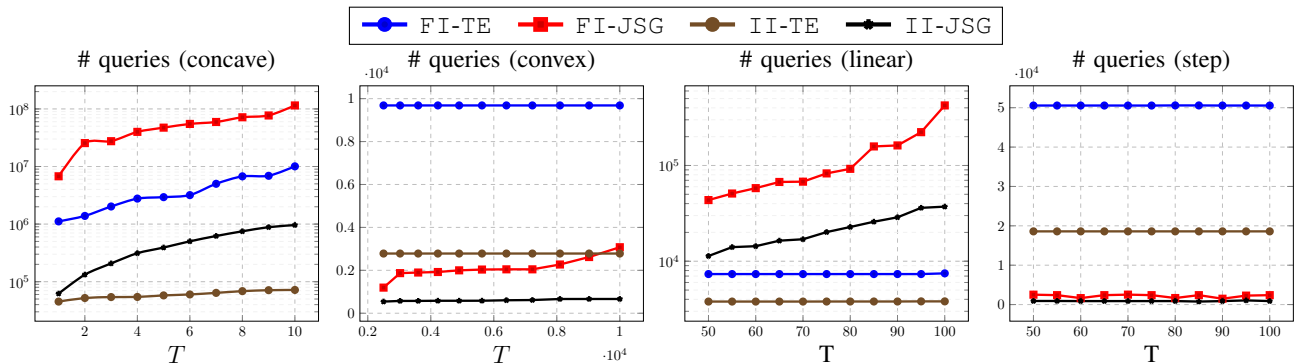
\begin{figure*}[t]
\begin{tikzpicture}[yscale=0.55, xscale=0.55]
    \begin{groupplot}[group style={group size= 4 by 1}]
        \nextgroupplot[title={\# queries (concave)},xlabel={$T$}, ymode=log, grid style=dashed, grid=both, grid style={line width=.1pt, draw=gray!10}, major grid style={line width=.2pt,draw=gray!50}, every axis plot/.append style={ultra thick, smooth},title style={font=\LARGE}, label style={font=\LARGE}]
                \addplot table [x=T, y=IC-TE, col sep=comma] {data/concave_T/query.csv};
                \addplot table [x=T, y=IC-JSG, col sep=comma] {data/concave_T/query.csv};
                \addplot table [x=T, y=II-TE, col sep=comma] {data/concave_T/query.csv};
                \addplot table [x=T, y=II-JSG, col sep=comma] {data/concave_T/query.csv};
                \coordinate (top) at (rel axis cs:0,1);
        \nextgroupplot[title={\# queries (convex)},xlabel={$T$}, grid style=dashed, grid=both, grid style={line width=.1pt, draw=gray!10}, major grid style={line width=.2pt,draw=gray!50}, every axis plot/.append style={ultra thick},title style={font=\LARGE}, label style={font=\LARGE}]
                \addplot table [x=T, y=IC-TE, col sep=comma] {data/convex_T/query.csv};
                \addplot table [x=T, y=IC-JSG, col sep=comma] {data/convex_T/query.csv};
                \addplot table [x=T, y=II-TE, col sep=comma] {data/convex_T/query.csv};
                \addplot table [x=T, y=II-JSG, col sep=comma] {data/convex_T/query.csv};
        \nextgroupplot[title={\# queries (linear)},xlabel={T}, ymode=log, grid style=dashed, grid=both, grid style={line width=.1pt, draw=gray!10}, major grid style={line width=.2pt,draw=gray!50}, every axis plot/.append style={ultra thick},title style={font=\LARGE}, label style={font=\LARGE}]
                \addplot table [x=T, y=IC-TE, col sep=comma] {data/linear_T/query.csv};
                \addplot table [x=T, y=IC-JSG, col sep=comma] {data/linear_T/query.csv};
                \addplot table [x=T, y=II-TE, col sep=comma] {data/linear_T/query.csv};
                \addplot table [x=T, y=II-JSG, col sep=comma] {data/linear_T/query.csv};
        \nextgroupplot[title={\# queries (step)},xlabel={T}, grid style=dashed, grid=both, grid style={line width=.1pt, draw=gray!10}, major grid style={line width=.2pt,draw=gray!50}, every axis plot/.append style={ultra thick},title style={font=\LARGE}, label style={font=\LARGE}]
                \addplot table [x=T, y=IC-TE, col sep=comma] {data/step_T/query.csv};
                \addplot table [x=T, y=IC-JSG, col sep=comma] {data/step_T/query.csv};
                \addplot table [x=T, y=II-TE, col sep=comma] {data/step_T/query.csv};
                \addplot table [x=T, y=II-JSG, col sep=comma] {data/step_T/query.csv};
                \coordinate (bot) at (rel axis cs:1,0);
    \end{groupplot}
\path (top|-current bounding box.north)--
      coordinate(legendpos)
      (bot|-current bounding box.north);
\matrix[
    matrix of nodes,
    anchor=south,
    draw,
    inner sep=0.2em,
    draw,
    nodes={anchor=center},
  ]at([xshift=-3cm]legendpos)
  {
    \ref{plot:ic_te}& \fsi-\te&[5pt]
    \ref{plot:ic_jsg}& \fsi-\jsg&[5pt]
    \ref{plot:ii_te}& \ii-\te&[5pt]
    \ref{plot:ii_jsg}& \ii-\jsg\\};
\end{tikzpicture}
\caption{Algorithms' number of queries with various $T$}
 	\label{fig:query_T}
\end{figure*}

\subsubsection{How the number of stored paths change?}

In the next experiment, we compare how much memory our algorithms took to process a \cspi instance. Feasible paths are critical to determine feasibility of our solution. An obstacle on preventing us to apply traditional constraint optimization on \cspi is listing all feasible paths, which could be exponential and a huge burden to computing storage. Therefore, we measures the memory efficiency of our algorithms in term of number of paths they need to store in memory in order to find a feasible solution. Figure \ref{fig:no_path} shows two kinds of charts of comparison between our algorithm: (1) One shows the maximum number of stored paths of each algorithms with various $T$; (2) The other one shows how the number of stored paths changes after each round of each algorithm. A round of my algorithm is counted as one while iteration of checking feasibility of obtained solutions.

From how \ii works, it is trivial that algorithms involving \ii store at most $O(|S|k)$ paths no matter value of $T$ is. That is also shown in Fig. \ref{fig:no_path}. On the other hand, the number of stored paths of algorithms involving \fsi increases when $T$ increases and is always much larger than this number in \ii. To have more insight, we look at how each algorithm accumulates paths after each round. As \fsi works in the manner that collects all feasible paths with unsatisfactory lengths in each round into one large set of paths, its number of paths starts from $O(|S|k)$ (the same as \ii) and increases significantly with more and more rounds to come. On the other hand, each round of \ii stores at most $O(|S|k)$ feasible paths; its path set in each round is disjoint and decreases in size. Therefore, \ii clearly shows its dominance to \fsi in term of memory.

Similar to the number of queries for finding global optimum of a univariate function, in linear cases, the number of stored paths of algorithms involving \te also stays constant and does not affected by value of $T$. The same reason is also applied.
 
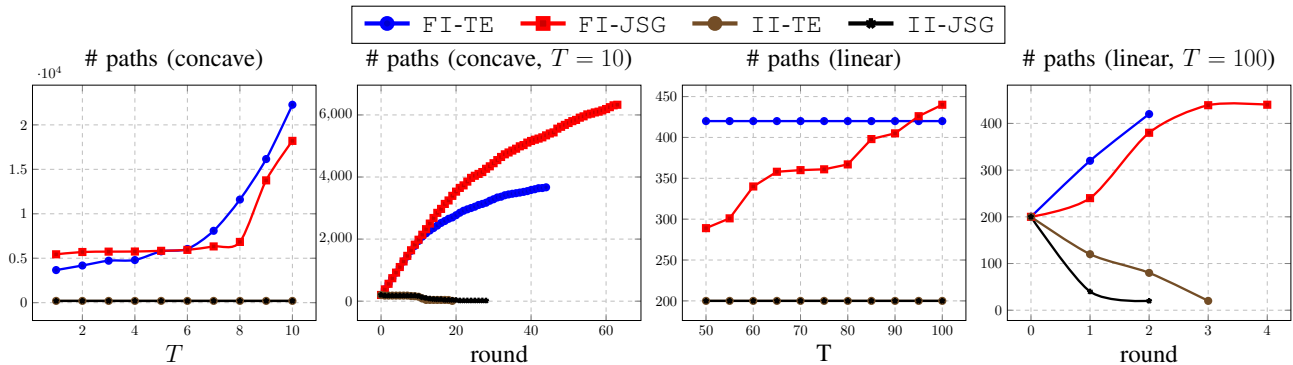
\begin{figure*}[t]
\begin{tikzpicture}[yscale=0.55, xscale=0.55]
    \begin{groupplot}[group style={group size= 4 by 1}]
        \nextgroupplot[title={\# paths (concave)},xlabel={$T$}, grid style=dashed, grid=both, grid style={line width=.1pt, draw=gray!10}, major grid style={line width=.2pt,draw=gray!50}, every axis plot/.append style={ultra thick, smooth},title style={font=\LARGE}, label style={font=\LARGE}]
                \addplot table [x=T, y=IC-TE, col sep=comma] {data/concave_T/path.csv};
                \addplot table [x=T, y=IC-JSG, col sep=comma] {data/concave_T/path.csv};
                \addplot table [x=T, y=II-TE, col sep=comma] {data/concave_T/path.csv};
                \addplot table [x=T, y=II-JSG, col sep=comma] {data/concave_T/path.csv};
                \coordinate (top) at (rel axis cs:0,1);
        \nextgroupplot[title={\# paths (concave, $T=10$)},xlabel={round}, grid style=dashed, grid=both, grid style={line width=.1pt, draw=gray!10}, major grid style={line width=.2pt,draw=gray!50}, every axis plot/.append style={ultra thick},title style={font=\LARGE}, label style={font=\LARGE}]
                \addplot table [x=round, y=IC-TE, col sep=comma] {data/concave_T/path_10_00.csv};
                \addplot table [x=round, y=IC-JSG, col sep=comma] {data/concave_T/path_10_00.csv};
                \addplot table [x=round, y=II-TE, col sep=comma] {data/concave_T/path_10_00.csv};
                \addplot table [x=round, y=II-JSG, col sep=comma] {data/concave_T/path_10_00.csv};
        \nextgroupplot[title={\# paths (linear)},xlabel={T}, grid style=dashed, grid=both, grid style={line width=.1pt, draw=gray!10}, major grid style={line width=.2pt,draw=gray!50}, every axis plot/.append style={ultra thick},title style={font=\LARGE}, label style={font=\LARGE}]
                \addplot table [x=T, y=IC-TE, col sep=comma] {data/linear_T/path.csv};
                \addplot table [x=T, y=IC-JSG, col sep=comma] {data/linear_T/path.csv};
                \addplot table [x=T, y=II-TE, col sep=comma] {data/linear_T/path.csv};
                \addplot table [x=T, y=II-JSG, col sep=comma] {data/linear_T/path.csv};
        \nextgroupplot[title={\# paths (linear, $T=100$)},xlabel={round}, grid style=dashed, grid=both, grid style={line width=.1pt, draw=gray!10}, major grid style={line width=.2pt,draw=gray!50}, every axis plot/.append style={ultra thick, smooth},title style={font=\LARGE}, label style={font=\LARGE}]
                \addplot table [x=round, y=IC-TE, col sep=comma] {data/linear_T/path_100_00.csv};
                \addplot table [x=round, y=IC-JSG, col sep=comma] {data/linear_T/path_100_00.csv};
                \addplot table [x=round, y=II-TE, col sep=comma] {data/linear_T/path_100_00.csv};
                \addplot table [x=round, y=II-JSG, col sep=comma] {data/linear_T/path_100_00.csv};
                \coordinate (bot) at (rel axis cs:1,0);
    \end{groupplot}
\path (top|-current bounding box.north)--
      coordinate(legendpos)
      (bot|-current bounding box.north);
\matrix[
    matrix of nodes,
    anchor=south,
    draw,
    inner sep=0.2em,
    draw,
    nodes={anchor=center},
  ]at([xshift=-3cm]legendpos)
  {
    \ref{plot:ic_te}& \fsi-\te&[5pt]
    \ref{plot:ic_jsg}& \fsi-\jsg&[5pt]
    \ref{plot:ii_te}& \ii-\te&[5pt]
    \ref{plot:ii_jsg}& \ii-\jsg\\};
\end{tikzpicture}
\caption{Algorithms' memory changes with various $T$}
 	\label{fig:no_path}
\end{figure*}

\subsubsection{Trade-off in term of $\varepsilon$}

In the final experiment, we investigate how different values of $\varepsilon$ impact our algorithms' performance. $\varepsilon$ represents how ``accurate" the returned solutions of our algorithms are to the requirement of \cspi. Intuitively, the smaller $\varepsilon$ is, the more accurate the solutions are, the closer lower bounds of distances between pairs of nodes on $S$ are to $T$. Fig. \ref{fig:solution_query_var} shows how our algorithms' returned solutions, their numbers of queries and stored paths change with various $\epsilon$.

From Fig. \ref{fig:solution_query_var}, we can see that the algorithm's returned impact amounts decrease with larger $\varepsilon$. This is intuitive since with more relaxed constraint, a smaller impact amount suffices. That is also reflected in our algorithms' theoretical approximation guarantee, in a way that the ratio is proportional to a term of $\ln \varepsilon^{-1}$. 

Beside the trade-off between solution accuracy and solution size, $\varepsilon$ also shows changes in the number of queries and stored paths of each algorithm. With algorithms involves \ii, large $\varepsilon$ helps decreasing number of queries, which totally contrasts with the one with \fsi. The behavior of \ii with various $\varepsilon$ is intuitively explained by the fact that: with a same path set, the more relaxed constraint should end up with the smaller overall impact needed. However, we found this fact does not applied with \fsi because the more relaxed constraint does not guarantee the fewer number of processed paths. That is shown in the third sub-figure in Fig. \ref{fig:solution_query_var}; we can see that the number of stored paths of \fsi increases with $\varepsilon$ grows. With more paths to process, \fsi's behavior becomes more complicated. Meanwhile, \ii is stable with the cap on the number of processing paths, which is at most $O(|S|k)$. 

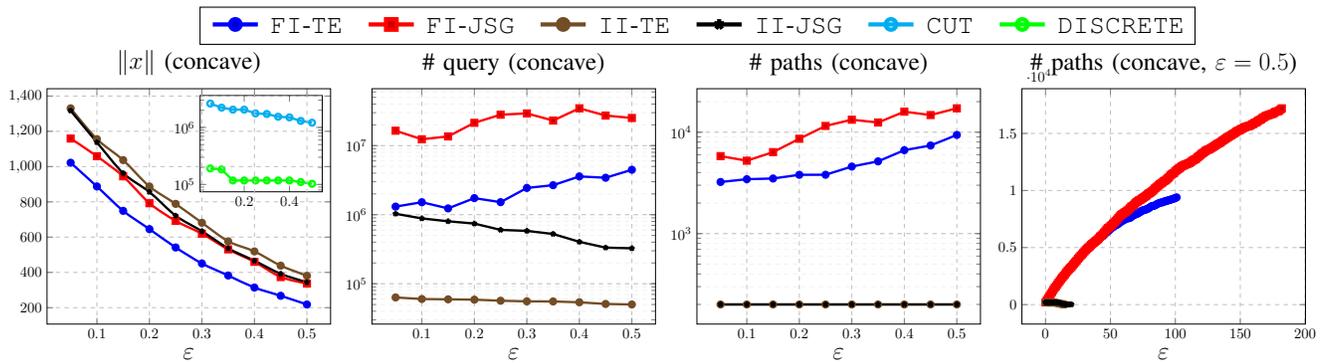
\begin{figure*}[t]
\begin{tikzpicture}[yscale=0.55, xscale=0.55]
    \begin{groupplot}[group style={group size= 4 by 1}]
        \begin{axis}[
                width=0.25\textwidth,ymode=log,
                xshift=3.7cm,
                yshift=3.2cm,
                every axis plot/.append style={ultra thick}
        ]
                \addplot[cyan,mark=o] table [x=var, y=CUT, col sep=comma] {data/concave_var/solution.csv};
                \addplot[green,mark=o] table [x=var, y=DISCRETE, col sep=comma] {data/concave_var/solution.csv};
        \end{axis}
        \nextgroupplot[title={$\norm{x}$ (concave)},xlabel={$\varepsilon$}, grid style=dashed, grid=both, grid style={line width=.1pt, draw=gray!10}, major grid style={line width=.2pt,draw=gray!50}, every axis plot/.append style={ultra thick, smooth},title style={font=\LARGE}, label style={font=\LARGE}]
                \addplot table [x=var, y=IC-TE, col sep=comma] {data/concave_var/solution.csv};
                \addplot table [x=var, y=IC-JSG, col sep=comma] {data/concave_var/solution.csv};
                \addplot table [x=var, y=II-TE, col sep=comma] {data/concave_var/solution.csv};
                \addplot table [x=var, y=II-JSG, col sep=comma] {data/concave_var/solution.csv};
                \coordinate (top) at (rel axis cs:0,1);
        \nextgroupplot[title={\# query (concave)},xlabel={$\varepsilon$}, ymode=log, grid style=dashed, grid=both, grid style={line width=.1pt, draw=gray!10}, major grid style={line width=.2pt,draw=gray!50}, every axis plot/.append style={ultra thick},title style={font=\LARGE}, label style={font=\LARGE}]
                \addplot table [x=var, y=IC-TE, col sep=comma] {data/concave_var/query.csv};
                \addplot table [x=var, y=IC-JSG, col sep=comma] {data/concave_var/query.csv};
                \addplot table [x=var, y=II-TE, col sep=comma] {data/concave_var/query.csv};
                \addplot table [x=var, y=II-JSG, col sep=comma] {data/concave_var/query.csv};
        \nextgroupplot[title={\# paths (concave)},xlabel={$\varepsilon$}, ymode=log, grid style=dashed, grid=both, grid style={line width=.1pt, draw=gray!10}, major grid style={line width=.2pt,draw=gray!50}, every axis plot/.append style={ultra thick},title style={font=\LARGE}, label style={font=\LARGE}]
                \addplot table [x=var, y=IC-TE, col sep=comma] {data/concave_var/path.csv};
                \addplot table [x=var, y=IC-JSG, col sep=comma] {data/concave_var/path.csv};
                \addplot table [x=var, y=II-TE, col sep=comma] {data/concave_var/path.csv};
                \addplot table [x=var, y=II-JSG, col sep=comma] {data/concave_var/path.csv};
        \nextgroupplot[title={\# paths (concave, $\varepsilon = 0.5$)},xlabel={$\varepsilon$}, grid style=dashed, grid=both, grid style={line width=.1pt, draw=gray!10}, major grid style={line width=.2pt,draw=gray!50}, every axis plot/.append style={ultra thick},title style={font=\LARGE}, label style={font=\LARGE}]
                \addplot table [x=round, y=IC-TE, col sep=comma] {data/concave_var/path_0_50.csv};
                \addplot table [x=round, y=IC-JSG, col sep=comma] {data/concave_var/path_0_50.csv};
                \addplot table [x=round, y=II-TE, col sep=comma] {data/concave_var/path_0_50.csv};
                \addplot table [x=round, y=II-JSG, col sep=comma] {data/concave_var/path_0_50.csv};
                \coordinate (bot) at (rel axis cs:1,0);
    \end{groupplot}
\path (top|-current bounding box.north)--
      coordinate(legendpos)
      (bot|-current bounding box.north);
\matrix[
    matrix of nodes,
    anchor=south,
    draw,
    inner sep=0.2em,
    draw,
    nodes={anchor=center},
  ]at([xshift=-3cm]legendpos)
  {
    \ref{plot:ic_te}& \fsi-\te&[5pt]
    \ref{plot:ic_jsg}& \fsi-\jsg&[5pt]
    \ref{plot:ii_te}& \ii-\te&[5pt]
    \ref{plot:ii_jsg}& \ii-\jsg&[5pt]
    \ref{plot:cut}& \texttt{CUT}&[5pt]
    \ref{plot:discrete}& \texttt{DISCRETE}\\};
\end{tikzpicture}
\caption{Trade-off in term of $\varepsilon$}
 	\label{fig:solution_query_var}
\end{figure*}

\subsection{Experiment Summary}

We summarize experimental results, showing advantages of our algorithms as follows:
\begin{itemize}
    \item Our algorithms outperform existing methods that needs an intermediate discretization step in most cases. Even in the special instance of \cspi with ``discrete" (step) function, one of our algorithm (\fsi-\jsg) performed comparably to the state-of-the-art solution. 
    \item Each of our algorithm has strengths in different aspects, to be specific:
    \begin{itemize}
        \item With the \tb oracle, algorithms involving \jsg tend to get better solution quality. Meanwhile, the ones with \te have advantage in the number of queries on global optimum of a univariate function.
        \item With the \cpl oracle, \fsi has strengths in solution quality while \ii shows to save memory in term of the number of stored feasible paths, which plays a role on saving the number of queries in the \tb oracle as well.
    \end{itemize}
    \item $\varepsilon$ allows user control the trade-off between solution quality and accuracy to the input constraint. Moreover, algorithms involving \ii benefit from $\varepsilon$ in the way that larger $\varepsilon$ helps reduce their runtime.
\end{itemize}

\section{Conclusion} \label{sec:conclusion}

We studied the \cspi problem, in which we proposed multiple algorithms with different performance guarantees. Theoretical evaluation and experimental analysis are provided, supporting users on deciding which combinations are the best for their needs. Indeed, there are still significant works to improve in the future. A node could be associated with multiple functions, serving for multiple objectives of system's functionality. Also, each function can have multiple variables and each variable could appear on more than one function, making the problem become much more complicated. How to balance those multiple objectives is still an open problem. 


\bibliographystyle{IEEEtran}
\bibliography{sample-bibliography}

\begin{IEEEbiography}[{\includegraphics[width=1in,height=1.25in,clip,keepaspectratio]{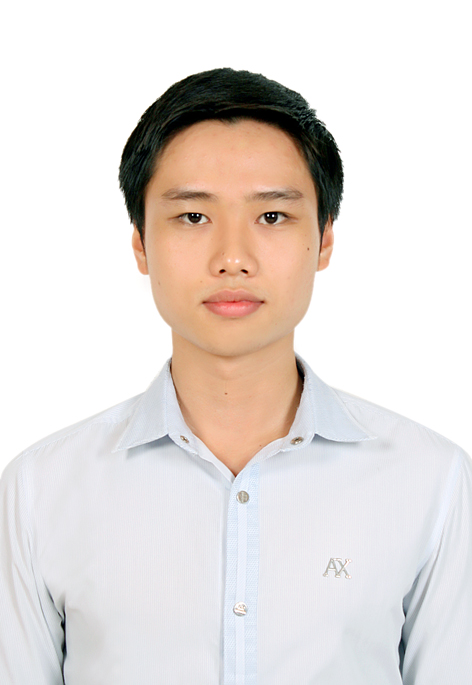}}]{Lan N. Nguyen} received his Degree of Engineer in Information Technology from Hanoi University of Science and Technology, Vietnam in 2014. He has been a PhD student under the supervisor of Dr. My T. Thai in the CISE department at the University of Florida since Spring 2017. His current research interests is on proposing lightweight algorithms to solve large-scale problems with application on Machine Learning or Network Optimization.
\end{IEEEbiography}

\begin{IEEEbiography}[{\includegraphics[width=1in,height=1.25in,clip,keepaspectratio]{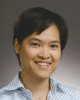}}] {My T. Thai}(M'06)
is a UF Research Foundation Professor at the Computer and Information Science and Engineering department, University of Florida. Her current research interests are on scalable algorithms, big data analysis, cybersecurity, and optimization in network science and engineering, including communication networks, smart grids, social networks, and their interdependency. The results of her work have led to 6 books and 170+ articles, including IEEE MSN 2014 Best Paper Award, 2017, IEEE ICDM Best Papers Award, 2017 IEEE ICDCS Best Paper Nominee, and 2018 IEEE/ACM ASONAM Best Paper Runner up. 

Prof. Thai has engaged in many professional activities. She has been a TPC-chair for many IEEE conferences, has served as an associate editor for \textit{IEEE Transactions} on \textit{Parallel and Distributed Systems}, \textit{IEEE Transactions} on \textit{Network Science} and \textit{Engineering}, and a series editor of \textit{Springer Briefs} in \textit{Optimization}. She is a founding Editor-in-Chief of the \textit{Computational Social Networks} journal, and Editor-in-Chief of \textit{Journal} of \textit{Combinatorial Optimization} (\textit{JOCO}). She has received many research awards including a UF Provosts Excellence Award for Assistant Professors, UFRF Professorship Award, a Department of Defense (DoD) Young Investigator Award, and an NSF (National Science Foundation) CAREER Award.
\end{IEEEbiography}


\end{document}